\documentclass[12pt]{article}%
\usepackage{amsfonts}
\usepackage{amsmath}
\usepackage{amssymb}
\usepackage{graphicx}
\usepackage{amsmath}%
\providecommand{\U}[1]{\protect\rule{.1in}{.1in}}
\newenvironment{proof}[1][Proof]{\noindent\textbf{#1.} }{\ \rule{0.5em}{0.5em}}
\providecommand{\U}[1]{\protect\rule{.1in}{.1in}}
\newtheorem{theorem}{Theorem}
\newtheorem{corollary}{Corollary}
\newtheorem{problem}{Problem}
\newtheorem{remark}{Remark}
\newtheorem{example}{Example}
\newtheorem{definition}{Definition}
\begin{document}

\title{\vspace*{-10mm}Deforming Lie algebras to Frobenius integrable non-autonomous
Hamiltonian systems}
\author{Maciej B\l aszak$^{1}$, Krzysztof Marciniak$^{2}$, Artur Sergyeyev$^{3}$\\$^{1}$Faculty of Physics, Division of Mathematical Physics,\\A. Mickiewicz University, Pozna\'n, Poland\\$^{2}$Department of Science and Technology, Campus Norrk\"oping,\\Link\"oping University, Sweden\\$^{3}$Mathematical Institute, Silesian University in Opava,\\Na Rybn\'\i{}\v{c}ku 1, 74601 Opava, Czech Republic\\E-mail: \texttt{blaszakm@amu.edu.pl, krzma@itn.liu.se,}\\\texttt{artur.sergyeyev@math.slu.cz} }
\maketitle

\begin{abstract}
Motivated by the theory of Painlev\'{e} equations and associated hierarchies,
we study non-autonomous Hamiltonian systems that are Frobenius integrable. We
establish sufficient conditions under which a given finite-dimensional Lie
algebra of Hamiltonian vector fields can be deformed to a time-dependent Lie
algebra of Frobenius integrable vector fields spanning the same distribution
as the original algebra. The results are applied to quasi-St\"{a}ckel systems
\cite{mb}.\looseness=-1


\end{abstract}



\section{Introduction}

The importance of the role played by integrable systems is hard to
overestimate, given both their manifold applications and their profound
connections to a number of areas in pure mathematics, see
e.g.\ \cite{b98,d,rs} and references therein. In particular,
finite-dimensional integrable Hamiltonian dynamical systems are well
understood, and the key tool in their study
is the Liouville theorem \cite{l} relating integrability to existence of
sufficiently many independent integrals of motion in involution.

This beautiful and well-studied setup involves a blanket assumption that the
systems under study do not involve explicit dependence on the evolution
parameter, i.e., time. Allowing for such an explicit dependence is far from
trivial and necessitates certain
nontrivial modifications of the very notion of integrability, see
e.g.\ \cite{aspla} and references therein for details. It should be pointed
out that the research in this subfield is relatively scarce compared with that
centered around the integrable dynamical systems in the setting of the
Liouville theorem, cf.\ e.g.\ \cite{b98,d, g, mr} and references therein.

However, there is an important motivation for the study of explicitly
time-dependent dynamical systems and their integrability: the Painlev\'e
equations, which play an important role in many areas of modern mathematics
and in applications, cf.\ e.g.\ \cite{J, n} and references therein, as well as
certain natural generalizations thereof \cite{bcr}, can be written as
time-dependent dynamical systems, see e.g.\ \cite{bcr, J, oka, t}.

In the present paper we take an approach to the study of time-dependent
dynamical systems and their integrability that, to the best of our knowledge,
was not systematically explored in the earlier literature. Namely, we
simultaneously consider several vector fields and the associated dynamical
systems, each with its own time, while allowing for an explicit dependence of
all vector fields on all times at once and imposing the Frobenius
integrability condition guaranteeing local existence of associated multitime
solutions, as explained below. In order to construct such vector fields
depending on all times at once, we begin with a Lie algebra of
time-independent vector fields on the underlying manifold and look for the
multiparameter deformations of this algebra having the desired properties.

We expect that this approach, possibly supplemented by certain additional
assumptions, will yield new nonautonomous Painlev\'{e}-type dynamical systems.

We start by the following definition of \emph{Frobenius integrability} that is
naturally motivated by an important notion of an integrable distribution and
by the Frobenius theorem from differential geometry, cf.\ e.g.\ \cite{Fecko,
Lundell, mr, o} and references therein.

\begin{definition}
A set of $n$ non-autonomous vector fields $Y_{i}(t_{1},\ldots,t_{n})$, each
depending on $n$ parameters $t_{i}$, on a finite-dimensional manifold $M$ is
called \emph{Frobenius integrable} if the following zero-curvature condition
(Frobenius condition) holds:%
\begin{equation}
\frac{\partial Y_{i}}{\partial t_{j}}-\frac{\partial Y_{j}}{\partial t_{i}%
}-\left[  Y_{i},Y_{j}\right]  =0\text{\ for all }i,j=1,\ldots,n, \label{genFr}%
\end{equation}
where $[\cdot,\cdot]$ stands for the Lie bracket (commutator) of vector fields.
\end{definition}

It is rather straightforward, cf.\ e.g.\ \cite{Fecko,J}, to see that if the
Frobenius condition (\ref{genFr}) is satisfied then the associated set of $n$
dynamical systems on $M$%
\begin{equation}
\frac{dx^{\alpha}}{dt_{i}}=Y_{i}^{\alpha}(\xi,t_{1},\dots,t_{n}),\quad
\alpha=1,\dots,m=\dim M\text{, \ \ }i=1,\dots,n, \label{sys2}%
\end{equation}
possesses a local common multi-time solution $x^{\alpha}=x^{\alpha}%
(t_{1},\dots,t_{n},\xi_{0})$ for each point $\xi_{0}\in M$, i.e., for each
initial condition $x^{\alpha}(0,\dots,0,\xi_{0})=x_{0}^{\alpha}$. Here
$x^{\alpha}$ are local coordinates on $M$ on a neighborhood of a point
$\xi_{0}$, and $x_{0}^{\alpha}$ are coordinates of $\xi_{0}$ in this
coordinate system; $\xi\in M$ denotes a point on $M$ and $Y^{\alpha}(\xi
,t_{1},\ldots,t_{n})$ is the value of $\alpha$-th component of the vector
field $Y$ w.r.t.\ local coordinate system given by $x^{\alpha}$ at the point
$\xi$ at the times $t_{1},\ldots,t_{n}$. Under obvious technical assumptions
such solutions from a set of overlapping local coordinate systems that can be
glued together to define an integral submanifold $\xi=\xi(\xi_{0},t_{1}%
,\dots,t_{n})$ passing through $\xi_{0}$. Such a submanifold gives us a
natural coordinate-free representation for the solution of the system in question.

Note that equations (\ref{genFr}) formally look exactly like the
zero-curvature-type equations arising in the study of integrable partial
differential dispersionless systems with Lax operators written in terms of
vector fields, cf.\ e.g.\ \cite{d, ms, as} and references therein. On the
other hand, if one of the times $t_{i}$ is identified with the variable
spectral parameter and the vector fields are replaced by matrices, then
equations (\ref{genFr}) formally look like the isomonodromic representations
for the Painlev\'e and Painlev\'e-type systems, cf.\ e.g.\ \cite{bcr} and
references therein.

Suppose now that $M$ is endowed with a nondegenerate Poisson structure $\pi$,
so we have a Poisson manifold $(M,\pi)$, and the vector fields $Y_{i}$ are
Hamiltonian, that is, $Y_{i}=\pi dH_{i}$ for some Hamiltonian functions
$H_{i}$ depending explicitly, in general, on all times $t_{k}$: $H_{i}%
=H_{i}(\xi,t_{1},\dots,t_{n})$. We stress that in our setup the Poisson
structure $\pi$ does not depend on any of $t_{k}$.\looseness=-1

As by definition of the Poisson bracket $\{\cdot,\cdot\}$ associated with
$\pi$ we have $\left[  \pi dH_{i},\pi dH_{j}\right]  =-\pi d\left\{
H_{i},H_{j}\right\}  $ (we use the sign convention $\left\{  H_{i}%
,H_{j}\right\}  =\left\langle dH_{i},\pi dH_{j}\right\rangle $, where
$\left\langle \cdot,\cdot\right\rangle $ is the natural pairing among
$T_{z}^{\ast}M$ and $T_{z}M$, although the opposite sign convention also
occurs in the literature), we immediately obtain that for (\ref{genFr}) to
hold it suffices that the following zero-curvature condition (Frobenius
condition) for Hamiltonians $H_{k}$ holds:
\begin{equation}
\frac{\partial H_{i}}{\partial t_{j}}-\frac{\partial H_{j}}{\partial t_{i}%
}+\{H_{i},H_{j}\}=0,\quad i,j=1,\dots,n \label{zcr}%
\end{equation}
In the case when the vector fields $Y_{i}$ do not depend explicitly on the
times $t_{i}$ the conditions (\ref{genFr}) and (\ref{zcr}) reduce to $\left[
Y_{i},Y_{j}\right]  =0$ \ for all $i,j$ and $\{H_{i},H_{j}\}=0$ for all $i,j$,
respectively. In such a case the vector fields $Y_{i}$ span an involutive (and
thus integrable by the Frobenius theorem \cite{Fecko}, \cite{Lundell})
distribution $\mathcal{D}$ while the Hamiltonian vector fields constitute a
Liouville intgrable system under certain additional regularity conditions.\looseness=-1

Recall that the members of the hierarchy associated with a given
Painlev\'{e} equation admit non-autonomous Hamiltonian formulations with
evolution parameters $t_{j}$ and explicitly time-dependent Hamiltonians
that satisfy (\ref{zcr}), cf.\ e.g.\ \cite{k, oka, t}. This suggests that,
conversely, some of the dynamical systems with the Hamiltonians that satisfy
(\ref{zcr}) could possess the Painlev\'{e} property but we defer the
investigation of this idea in more detail to future work.\looseness=-1

Motivated by the above, in the present paper we study existence of
polynomial-in-times deformations of Lie algebras of autonomous Hamiltonians
$h_{i}$ (so the associated Hamiltonian vector fields $X_{i}=\pi dh_{i}$
satisfy (\ref{inv}) with $c_{ij}^{k}$ being constants) such that the deformed
Hamiltonians $H_{i}$ satisfy the condition (\ref{zcr}). In this way we
produce, from the system of non-commuting autonomous vector fields $X_{i}=\pi
dh_{i}$, polynomial-in-times vector fields $Y_{i}$ that satisfy (\ref{genFr}),
which guarantees existence of common multi-time solutions for the set of
non-autonomous systems (\ref{sys2}). Under certain natural assumptions
this deformation is shown to be unique, see Theorem \ref{main} in Section 2
for details. Then, in Section~3, we apply our general theory to the so-called
quasi-St\"{a}ckel systems \cite{mb} and present a way of explicit computation
of the deformations in question in this particular setting. As a result, we
construct a number of families of
non-autonomous Hamiltonian systems with $n$ degrees of freedom integrable in
the Frobenius sense.\looseness=-1

\section{Non-autonomous deformations of Lie algebras yielding Frobenius
integrability}

Consider an $n$-dimensional ($1<n<\dim M$) Lie algebra $\mathfrak{g}%
=\mathrm{span}\{h_{i}\in C^{\infty}(M),i=1,\dots,n\}$ of smooth real-valued
functions on our Poisson manifold $\left(  M,\pi\right)  $, with the structure
constants $c_{ij}^{k}\in\mathbb{R}$ so that
\begin{equation}
\{h_{i},h_{j}\}=\sum_{k=1}^{n}c_{ij}^{k}h_{k} \label{alg}%
\end{equation}
is the Lie bracket on $\mathfrak{g}$, cf.\ Definition 6.41 in \cite{olver}. We
assume that the functions $h_{i}:M\rightarrow\mathbb{R}$ (the Hamiltonians)
are functionally independent. The functions $h_{i}$ define $n$ autonomous
conservative Hamiltonian systems, cf.\ (\ref{sys2}),
\begin{equation}
\frac{dx^{\alpha}}{dt_{i}}=\left(  \pi(\xi)dh_{i}(\xi,t_{1},\dots
,t_{n})\right)  ^{\alpha},\quad\alpha=1,\dots,m=\dim M,\text{ \ \ }%
i=1,\dots,n,\quad\label{hs}%
\end{equation}
on $M$, and the Hamiltonian vector fields $X_{i}=\pi dh_{i}$ satisfy
\begin{equation}
\left[  X_{i},X_{j}\right]  =-\sum_{i=1}^{n}c_{ij}^{k}X_{k} \label{inv}%
\end{equation}
and thus span an involutive, and hence integrable in the sense of Frobenius,
distribution $\mathcal{D}$ on $M$.

It is well known that if (\ref{inv}) holds then one can choose a basis
$V_{1},\dots,V_{n}$ of vector fields spanning the distribution $\mathcal{D}$
such that
$\left[  V_{i},V_{j}\right]  =0$ for all $i,j=1,\dots,n$,
cf.\ e.g.\ \cite{Lundell}\looseness=-1. However, a direct (explicit)
construction of such a basis is usually not possible and the basis $V_{i}$
does not have to consist of Hamiltonian vector fields. We would therefore like
to have a method of deforming, in a precise sense defined below, the
autonomous vector fields $X_{i}$ to (non-autonomous in general) vector fields
$Y_{i}$ such that

\begin{enumerate}
\item The vector fields $Y_{i}$ span the same distribution $\mathcal{D}$ as
$X_{i}$ do.

\item The vector fields $Y_{i}$ are Hamiltonian with respect to $\pi$ just as
$X_{i}$, so $Y_{i}=\pi dH_{i}$ for some functions $H_{i}$ depending in general
on all times $t_{i}$.

\item The dynamical systems (\ref{sys2}) defined by the vector fields $Y_{i}$
possess common muti-time solutions, so that the condition (\ref{genFr}) is
satisfied and (\ref{zcr}) is valid for $H_{i}$.

Mathematically, this problem can be stated as follows.
\end{enumerate}

\begin{problem}
\label{problem} Denote by $\mathfrak{g}[t_{1},\dots,t_{n}]$ the space of
multivariate polynomials in $t_{1},\dots,t_{n}$ with values in $\mathfrak{g}
$.\newline1. Can one find (and, if yes, under which conditions) nonzero
polynomials $H_{i}\in\mathfrak{g}[t_{1},\dots,t_{n}]$, $i=1,\dots,n$, such
that the non-autonomous Frobenius condition (\ref{zcr}) holds and such that
$Y_{i}=\pi dH_{i}$ span the same distribution $\mathcal{D}$ as $X_{i}$
do?\newline2. Is there a unique answer to question 1?\newline3. Is there an
explicit way to calculate $H_{i}$?
\end{problem}

Thus, we will look for polynomial-in-times deformations $H_{i}$ of the
Hamiltonians $h_{i}$ such that $\pi dH_{i}$ and $\pi dh_{i}$ span the same
distribution $\mathcal{D}$ and such that the non-autonomous Hamiltonian
systems
\begin{equation}
\frac{dx^{\alpha}}{dt_{i}}=\left(  \pi(\xi)dH_{i}(\xi,t_{1},\dots
,t_{n})\right)  ^{\alpha},\quad\alpha=1,\dots,m=\dim M,\text{ \ \ }%
i=1,\dots,n,\quad\label{nhs}%
\end{equation}
satisfy the Frobenius condition (\ref{genFr}) and thus possess common
multi-time solutions $\xi=\xi(t_{1},\dots,t_{n},\xi_{0})$.

The first two questions of Problem \ref{problem} can be answered in the
following general setting.


\begin{theorem}
\label{main}Suppose that in a finite-dimensional Lie algebra $\mathfrak{g}$
there exists a basis $\{h_{i}\}_{i=1}^{n}$ such that

\begin{itemize}
\item[i)] $\mathfrak{g}_{c}=\mathrm{span}\{h_{i}:i=1,\dots,d_{c}\}$, where
$d_{c}\geq1$, is the center of $\mathfrak{g}$, so that for any $i=1,\dots
,d_{c}$ we have $\{h_{i},h\}=0$ for any $h\in\mathfrak{g}$;

\item[ii)] $\mathfrak{g}_{a}=\mathrm{span}\{h_{i}:i=1,\dots,d_{a}\}$, where
$d_{a}\geq d_{c}$, is an Abelian subalgebra of $\mathfrak{g}$;\looseness=-1

\item[iii)] $\left\{  h_{i},h_{j}\right\}  \in\mathrm{span}\left\{
h_{1},\dots,h_{\min(i,j)-1}\right\}  $ for all $i,j\leq n-1$.
\end{itemize}

Then there exists a unique multi-time-dependent Lie algebra (multi-time formal
deformation of $\mathfrak{g}$) with the generators $H_{i}\in\mathfrak{g[}%
t_{d_{c}+1},\dots,t_{i-1}]$, $i=1,\dots,n\,$\ such that Frobenius
integrability conditions (\ref{zcr}) hold, provided that

\begin{itemize}
\item[a)] $H_{i}=h_{i}$, $i=1,\dots,d_{a}$,

\item[b)] $H_{i}|_{t_{d_{c}+1}=0,\dots,t_{i-1}=0}=h_{i}$, $i=d_{a}+1,\dots,n$.
\end{itemize}
\end{theorem}

The assumptions i)--iii) imply that $\mathfrak{g}_{c}\subset\mathfrak{g}
_{a}\subset\mathfrak{g}_{n-1}=\mathrm{span}\{h_{i}:i=1,\dots,n-1\}\subset
\nolinebreak\mathfrak{g}$. Moreover, iii) implies that $\mathfrak{g}_{n-1}$ is
a nilpotent subalgebra of the Lie algebra $\mathfrak{g}$ of codimension one.
The theorem of course encompasses the case when $\mathfrak{g}$ itself, rather
than just $\mathfrak{g}_{n-1}$, is nilpotent. Note also that thanks to the
assumption a) we have $H_{i}=h_{i}$ for $i=1,\dots,d_{a}$, so the Hamiltonians
$h_{i}$ spanning the Abelian subalgebra $\mathfrak{g}_{a}$ are not deformed.\looseness=-1

We stress that both the statement of Theorem~\ref{main} and its proof given
below are purely algebraic, so Theorem~\ref{main} holds not just for Lie
algebras of functions on a Poisson manifold but for an arbitrary
finite-dimensional Lie algebra $\mathfrak{g}$ which satisfies the conditions
of the theorem, with
(\ref{zcr}) replaced by
\[
\frac{\partial H_{i}}{\partial t_{j}}-\frac{\partial H_{j}}{\partial t_{i}%
}+[\![ H_{i},H_{j}]\!]=0,\quad i,j=1,\dots, n, \eqno{(5')}
\]
where $[\![\cdot,\cdot]\!]$ denotes the Lie bracket in $\mathfrak{g}$, and
$n=\dim\mathfrak{g}$. Then
in the proof the Poisson bracket $\{\cdot,\cdot\}$ should also be replaced by
$[\![\cdot,\cdot]\!]$.\looseness=-1

\begin{proof}
By virtue of a) we have that $H_{i}=h_{i}$ for $i=1,\dots, d_{a}$. Now, as
$\partial H_{j}/\partial t_{d_{a}+1}=0$ for $j=1,\dots,d_{a}$ by assumption,
the deformed Hamiltonian $H_{d_{a}+1}$ can be determined from the following
(part of) equations (\ref{zcr}):
\looseness=-1
\begin{equation}
\label{zcr0}\{H_{j},H_{d_{a}+1}\}-\frac{\partial H_{d_{a}+1}}{\partial t_{j}%
}=0,\quad j=1,\dots,d_{a}.
\end{equation}
This system has a (unique due to b)) solution
\begin{equation}
H_{d_{a}+1}=\exp\left(  -\sum\limits_{i=d_{c}+1}^{d_{a}}t_{i}\mathrm{ad}%
_{h_{i}}\right)  h_{q+1}, \label{a}%
\end{equation}
as we have $\mathrm{ad}_{h_{i}}=0$ for $i=1,\dots,d_{c}$; recall that by
definition $\mathrm{ad}_{f}(h)=\{f,h\}$ for any $f,h\in\mathfrak{g}$. Thus,
$H_{d_{a}+1}$ depends only on times $t_{d_{c}+1},\dots,t_{d_{a}}$. Note that
the expression in (\ref{a}) is a polynomial in $t_{d_{c}+1},\dots,t_{d_{a}}$
by virtue of the nilpotency assumption iii).\looseness=-1 For the remaining
$H_{i}$ we proceed by induction. Suppose that $H_{j}=H_{j}(t_{d_{c}+1}%
,\ldots,t_{j-1})$, $j=1,\dots, k$ are already known. Then $H_{k+1}$ can be
uniquely determined from equations (\ref{zcr}) which due to the fact that
$\partial H_{j}/\partial t_{k+1}\allowbreak=0$ for $j=1,\ldots, k$ read
\begin{equation}
\{H_{j},H_{k+1}\}-\frac{\partial H_{k+1}}{\partial t_{j}}=0,\quad j=1,\dots,k.
\label{h}%
\end{equation}
The first $d_{c}$ of these equations yield
\[
\frac{\partial H_{k+1}}{\partial t_{j}}=0,\quad j=1,\dots,d_{c},
\]
as for $j=1,\dots,d_{c}$ we have $\{H_{j},H_{k+1}\}=0$ by the assumption i).
This means that $H_{k+1}$ does not depend on $t_{1},\ldots, t_{d_{c}}$. The
remaining equations in (\ref{h}) therefore have a (unique due to b)) solution
of the form (cf.\ e.g.\ \cite{f} and references therein)
\begin{equation}
H_{k+1}=\mathcal{P}\exp\left(  -\int_{\gamma}\sum\limits_{i=d_{c}+1}%
^{k}\mathrm{ad}_{H_{i}}dt_{i}\right)  h_{k+1} \label{aa}%
\end{equation}
where $\gamma$ is any (smooth) curve in (an open domain of) $\mathbb{R}%
^{k-q_{0}}$ connecting the points $0$ and $(t_{d_{c}+1},\dots,t_{k})$ and
where $\mathcal{P} \exp$ denotes the path-ordered exponential, see
e.g.\ \cite{o} and references therein. This integral does not depend on a
particular choice of $\gamma$ because of the zero-curvature equations
(\ref{h}). Parameterizing the curve $\gamma$ by a parameter $\tau^{\prime}%
\in\left[  0,\tau\right]  $ so that $\gamma(0)=0$ and $\gamma(\tau
)=(t_{d_{c}+1},\dots,t_{k})$ yields
\begin{equation}
\label{F}\int_{\gamma}\sum\limits_{i=d_{c}+1}^{k}\mathrm{ad}_{H_{i}}%
dt_{i}=\int_{0}^{\tau}\sum\limits_{i=d_{c}+1}^{k}\mathrm{ad}_{H_{i}}%
|_{t_{j}=t_{j}(\tau^{\prime})}dt_{i}(\tau^{\prime})\equiv-\int_{0}^{\tau}%
F_{k}(\tau^{\prime})d\tau^{\prime}%
\end{equation}
\nopagebreak[4] where $F_{k}$ is an $\mathrm{End}(\mathfrak{g})$-valued
function of the parameter $\tau^{\prime}$.
The path-ordered exponential can be computed using the following formal Magnus
expansion, see e.g.\ \cite{o} and references therein:
\begin{equation}%
\begin{array}
[c]{l}%
\mathcal{P}\exp\left(  \displaystyle \int_{0}^{\tau}F_{k}(\tau^{\prime}%
)d\tau^{\prime}\right)  =\displaystyle\sum\limits_{s=0}^{\infty}\Omega_{s}%
^{k}\\[5mm]%
\displaystyle\equiv\sum\limits_{s=0}^{\infty}\displaystyle\frac{1}{s!}\int
_{0}^{\tau}d\tau_{1}^{\prime}\int_{0}^{\tau_{1}^{\prime}}\!d\tau_{2}^{\prime
}\cdots\int_{0}^{\tau_{s-2}^{\prime}}\!d\tau_{s-1}^{\prime}\int_{0}%
^{\tau_{s-1}^{\prime}}F_{k}(\tau_{1}^{\prime})\cdots F_{k}(\tau_{s}^{\prime
})d\tau_{s}^{\prime}.
\end{array}
\label{path0}%
\end{equation}
To complete the proof it remains to establish the polynomiality of $H_{k+1}$
in $t_{d_{c}+1},\dots, t_{k}$. This is achieved by observing that
$\Omega_{s}^{k}$ for all $k=d_{a}+1,\dots,n$ involve only $\mathrm{ad}_{H_{j}%
}$ with $j=d_{c}+1,\dots, n-1$ but do not involve $\mathrm{ad}_{H_{n}}$;
therefore $\Omega_{s}^{k}$ vanish for sufficiently large $s$ as the
expressions like
\[
\mathrm{ad}_{H_{r_{1}}}\,\mathrm{ad}_{H_{r_{2}}}\cdots\,\mathrm{ad}_{H_{r_{j}%
}}%
\]
will all vanish for sufficiently large $j$ if all $r_{i}$ belong to
$d_{c}+1,\dots,n-1$ as $\mathfrak{g}_{n-1}$ is nilpotent by virtue of
assumption iii).\looseness=-1
\end{proof}

Notice that the non-autonomous Hamiltonian systems (\ref{nhs}) are
conservative by construction, as the $i$-th Hamiltonian $H_{i}$ does not
depend on its own evolution parameter $t_{i}$. Moreover, for $k>j$ the
Frobenius conditions (\ref{zcr}) read\looseness=-1
\begin{equation}
\frac{\partial H_{k}}{\partial t_{j}}-\{H_{j},H_{k}\}=\frac{\partial H_{k}%
}{\partial t_{j}}+\{H_{k},H_{j}\}=\left(  \frac{\partial}{\partial t_{j}%
}+L_{Y_{j}}\right)  H_{k}=0,\qquad k>j \label{Li}%
\end{equation}
where $L_{Y_{j}}$ is the Lie derivative along the vector field $Y_{j}$, so all
$H_{k}$ with $k>j$ are time-dependent integrals of motion for the $j$-th flow.

\begin{remark}
Note that by the very construction of the Hamiltonians $H_{i}$ the vector
fields $Y_{i}=\pi dH_{i}$ span the same distribution $\mathcal{D}$ as
$X_{i}=\pi dh_{i}$, as required in part 1) of Problem \ref{problem}.\looseness=-1
\end{remark}

\section{Non-autonomous deformations\newline of quasi-St\"{a}ckel
Hamiltonians}

In this section we apply Theorem \ref{main} to quasi-St\"{a}ckel systems
constructed in \cite{mb}; cf.\ also e.g.\ \cite{b2005,bls2007,mb} for general
background on St\"ackel and quasi-St\"ackel systems. In this particular
setting we will be able to compute the expressions in (\ref{a}) and
(\ref{aa}), thus answering the question 3 of Problem \ref{problem}.

Fix an $n\in\mathbb{N}$, $n\geq2$. Consider a $2n$-dimensional Poisson
manifold $M$ and a particular set $(\lambda_{i},\mu_{i})$ of local Darboux
(canonical) coordinates on $M$, so that $\{\mu_{i},\lambda_{j}\}=\delta_{ij},$
$i,j=1,\dots,n$ while all $\{\lambda_{i},\lambda_{j}\}$ and $\{\mu_{i},\mu
_{j}\}$ are zero. Fix also $m\in\{0,\dots,n+1\}$ and consider the following
system of linear quasi-separation relations \cite{mb} (cf.\ also \cite{m})
\begin{equation}
\sum_{j=1}^{n}\lambda_{i}^{n-j}h_{j}=\frac{1}{2}\lambda_{i}^{m}\mu_{i}%
^{2}+\sum_{k=1}^{n}v_{ik}(\lambda)\mu_{k},\qquad i=1,\dots,n, \label{sep}%
\end{equation}
where
\[
\sum_{k=1}^{n}v_{ik}(\lambda)\mu_{k}=%
\begin{cases}
\displaystyle-\sum_{k\neq i}\frac{\mu_{i}-\mu_{k}}{\lambda_{i}-\lambda_{k}}, &
\text{for }m=0,\\[5mm]%
\displaystyle-\lambda_{i}^{m-1}\sum_{k\neq i}\frac{\lambda_{i}\mu_{i}%
-\lambda_{k}\mu_{k}}{\lambda_{i}-\lambda_{k}}+(m-1)\lambda_{i}^{m-1}\mu_{i}, &
\text{for }m=1,\dots,n,\\[5mm]%
\displaystyle-\lambda_{i}^{n-1}\sum_{k\neq i}\frac{\lambda_{i}^{2}\mu
_{i}-\lambda_{k}^{2}\mu_{k}}{\lambda_{i}-\lambda_{k}}+(n-1)\lambda_{i}^{n}%
\mu_{i}, & \text{for}\quad m=n+1.
\end{cases}
\]
Solving (\ref{sep}) with respect to $h_{j}$ yields, for each choice of
$m\in\{0,\dots,n+1\}$, $n$ Hamiltonians on $M$:
\[
h_{1}=E_{1}=\frac{1}{2}\mu^{T}G\mu,\quad h_{i}=E_{i}+W_{i},\quad i=2,\ldots,n
\]
where
\[
E_{i}=\frac{1}{2}\mu^{T}A_{i}\mu,\qquad W_{i}=\mu^{T}Z_{i},\quad
i=2,\ldots,n,
\]
are generated by the first respective second term on the right hand side of
(\ref{sep}) (we chose to omit the index $m$ in the above notation to simplify
writing). Here
\[
G=\text{diag}\left(  \frac{\lambda_{1}^{m}}{\Delta_{1}},\ldots,\frac
{\lambda_{n}^{m}}{\Delta_{n}}\right)  ,\quad\Delta_{i}=\prod\limits_{j\neq
i}(\lambda_{i}-\lambda_{j})
\]
can be interpreted as a contravariant metric tensor on an $n$-dimensional
manifold $Q$, $E_{1}$ can then be interpreted as the geodesic Hamiltonian of a
free particle in the pseudo-Riemaniann configuration space $(Q,g=G^{-1})$ so
that $M=T^{\ast}Q$ \cite{b2005,bls2007}. Next, $A_{r}=K_{r}G$, where $K_{r}$
are $(1,1)$-Killing tensors for metric $g$ with any chosen $m\in
\{0,\dots,n+1\}$, and are given by
\[
K_{i}=(-1)^{i+1}\text{diag}\left(  \frac{\partial\sigma_{i}}{\partial
\lambda_{1}},\dots,\frac{\partial\sigma_{i}}{\partial\lambda_{n}}\right)
\quad i=1,\dots,n
\]
where $\sigma_{r}(\lambda)$ are elementary symmetric polynomials in $\lambda$.
Moreover, $E_{i}$ are integrals of motion for $E_{1}$ as in fact they all
pairwise commute: $\{E_{i},E_{j}\}\nolinebreak=0,\ i,j=1,\ldots,n$. The vector
fields $Z_{i}$ are in this setting the Killing vectors of the metric $g$ for
any $m\in\{0,\dots,n+1\}$ as $L_{Z_{k}}g=0$, and they take the form \cite{mb}
\[
\left(  Z_{i}\right)  ^{\alpha}=\sum\limits_{k=1}^{i-1}(-1)^{i-k}%
\,k\,\sigma_{i-k-1}\frac{\lambda_{\alpha}^{m+k-1}}{\Delta_{\alpha}},\qquad
i\in I_{1}^{m}%
\]
and
\[
\left(  Z_{i}\right)  ^{\alpha}=\sum\limits_{k=1}^{n-i+1}(-1)^{i+k}%
\,k\,\sigma_{i+k-1}\frac{\lambda_{\alpha}^{m-k-1}}{\Delta_{\alpha}},\qquad
i\in I_{2}^{m}%
\]
where
\[
I_{1}^{m}=\{2,\dots,n-m+1\},\qquad I_{2}^{m}=\{n-m+2,\ldots,n\},\qquad
m=0,\ldots,n+1.
\]
Note that the above notation implies that%
\[
I_{1}^{0}=\{2,\dots,n\}\text{, }I_{1}^{n}=I_{1}^{n+1}=\emptyset\text{,
\ }I_{2}^{0}=I_{2}^{1}=\emptyset.
\]
It was demonstrated in \cite{mb} that the Hamiltonians $h_{i}$ constitute a
Lie algebra $\mathfrak{g}=\mathrm{span}\{h_{i}\in C^{\infty}(M)\colon
i=1,\ldots,n\}$ with the following commutation relations:
\[
\{h_{1},h_{i}\}=0,\quad i=2,\dots,n,
\]
and
\begin{equation}
\{h_{i},h_{j}\}=%
\begin{cases}
0, & \text{for }i\in I_{1}^{m}\text{ and }j\in I_{2}^{m},\\
(j-i)h_{i+j-(n-m+2)}, & \text{for }i,j\in I_{1}^{m},\\
-(j-i)h_{i+j-(n-m+2)}, & \text{for }i,j\in I_{2}^{m},
\end{cases}
\label{str}%
\end{equation}
where $i,j=2,\ldots,n$. We use here the convention that $h_{i}=0$ for $i\leq0$
or $k>n$.

\begin{remark}
The Lie algebra $\mathfrak{g}$ splits into a direct sum of Lie subalgebras
$\mathfrak{g=g}_{I_{1}}\mathfrak{\oplus g}_{I_{2}}$ where
\[
\mathfrak{g}_{I_{1}}=\mathrm{span}\{h_{1}\}\mathfrak{\oplus}\mathrm{span}%
\{h_{r}\colon r\in I_{1}^{m}\}\quad\mbox{and}\quad\mathfrak{g}_{I_{2}%
}=\mathrm{span}\{h_{r}\colon r\in I_{2}^{m}\}.
\]

\end{remark}

In order to successfully apply Theorem \ref{main} and formulas (\ref{a}) and
(\ref{aa}) we will now focus on the cases $m=0,1,$ when $\mathfrak{g=g}%
_{I_{1}}$, since $I_{2}^{m}$ is then empty. Note also that for these cases the
Lie algebra $\mathfrak{g}$ is nilpotent. Then (\ref{str}) reads%
\[
\mathrm{ad}_{h_{s_{1}}}h_{i}=\{h_{i},h_{s_{1}}\}=(i,s_{1})h_{i+s_{1}%
-(n-m+2)}\text{ with \thinspace}(i,s_{1})=s_{1}-i
\]
from which it immediately follows that for any $k\in\mathbb{N}$
\begin{equation}
\mathrm{ad}_{h_{s_{k}}}\cdots\mathrm{ad}_{h_{s_{1}}}h_{i}=(i,s_{1},\dots
,s_{k})h_{i+s_{1}+\ldots+s_{k}-k(n-m+2)} \label{1}%
\end{equation}
where
\begin{equation}
(i,s_{1}\dots,s_{k})=(i,s_{1},\dots,s_{k-1})[s_{k}-s_{k-1}-\cdots
-s_{1}-i+s(n-m+2)]. \label{2}%
\end{equation}
Note that in (\ref{1}) we put $h_{s}=0$ for $s<1$.

\begin{theorem}
\label{main2}Suppose that $m=0$ or $m=1$ (then $I_{2}^{m}$ is empty while
$d_{c}=2$ for $m=0$ and $d_{c}=1$ for $m=1$). Then the conditions of
Theorem~\ref{main} are satisfied and the polynomial-in-times deformation of
$\mathfrak{g}$ given by formulas (\ref{a}) and (\ref{aa}) can be written in
the form%
\begin{equation}
\hspace*{-5mm}%
\begin{array}
[c]{rcl}%
H_{i} & = & h_{i}-\displaystyle\!\!\!\sum_{r_{1}=d_{c}+1}^{i-1}\left(
\mathrm{ad}_{h_{r_{1}}}h_{i}\right)  t_{r_{1}}+\!\!\!\sum_{r_{1}=d_{c}%
+1}^{i-1}\sum_{r_{2}=r_{1}}^{i-1}\alpha_{ir_{1}r_{2}}\left(  \mathrm{ad}%
_{h_{r_{2}}}\mathrm{ad}_{h_{r_{1}}}h_{i}\right)  t_{r_{1}}t_{r_{2}}\\[7mm]
&  & -\displaystyle\sum_{r_{1}=d_{c}+1}^{i-1}\sum_{r_{2}=r_{1}}^{i-1}%
\sum_{r_{3}=r_{2}}^{i-1}\alpha_{ir_{1}r_{2}r_{3}}\left(  \mathrm{ad}%
_{h_{r_{3}}}\mathrm{ad}_{h_{r_{2}}}\mathrm{ad}_{h_{r_{1}}}h_{i}\right)
t_{r_{1}}t_{r_{2}}t_{r_{3}}+\cdots,
\end{array}
\label{3}%
\end{equation}
and the real constants $\alpha_{ir_{1}\cdots r_{k}}$ can be uniquely
determined from the Frobenius integrability condition (\ref{zcr}).
\end{theorem}

\begin{proof}
The formula (\ref{1}) implies that the center $\mathfrak{g}_{c}$ for $m=0$ is
two-dimensional and given by $\mathfrak{g}_{c}=\mathrm{span}\left\{
h_{1},h_{2}\right\}  $ while for $m=1$ the center $\mathfrak{g}_{0}$ is
one-dimensional and spanned by $h_{1}$ only. The same formula implies also
that in both cases $\{h_{i},h_{j}\}\in\mathrm{span}(h_{1},\dots,h_{\min
(i,j)-1})\,$\ for all $i,j=1,\dots,n$ so the conditions i)--iii) of
Theorem~\ref{main} are satisfied. The explicit form (\ref{3}) of deformations
(\ref{a}) and (\ref{aa}) is obtained by a direct computation using the
formulas given in the proof of Theorem~\ref{main} and taking a straight line
for $\gamma$.\looseness=-1
\end{proof}

Notice that from (\ref{str}) it follows that the dimension $d_{a}$ of the
Abelian subalgebra $\mathfrak{g}_{a}$ of $\mathfrak{g}$ is given by
\begin{equation}
d_{a}=\left[  \frac{n+3-m}{2}\right]  ,\quad m=0,1 \label{dimga}%
\end{equation}
so by Theorem \ref{main} the first $d_{a}$ Hamiltonians $h_{i}$ will not be
deformed, and that in (\ref{3}) $\ i=d_{a}+1,\dots,n$.

Theorem \ref{main2} gives us an effective way of calculating the sought-for
deformations, as it will be demonstrated in the following examples. Of course,
the highest order of polynomials in $t_{j}$ obtained in this way depends on
$n$.


\begin{example}
\bigskip Consider the case $n=6$, $m=0$. Then the formulas (\ref{str}) yield
the following matrix of commutators $\{h_{i},h_{j}\}$:
\[
\left(
\begin{array}
[c]{cccccc}%
0 & 0 & 0 & 0 & 0 & 0\\
0 & 0 & 0 & 0 & 0 & 0\\
0 & 0 & 0 & 0 & 0 & 3h_{1}\\
0 & 0 & 0 & 0 & h_{1} & 2h_{2}\\
0 & 0 & 0 & -h_{1} & 0 & h_{3}\\
0 & 0 & -3h_{1} & -2h_{2} & -h_{3} & 0
\end{array}
\right)  ,
\]
and clearly $d_{c}=2$ while $d_{a}=4$. The explicit values of the expansion
coefficients $\alpha_{ir_{1}\dots r_{k}}$ can be obtained by plugging
(\ref{3}) into (\ref{zcr}). Having done this we obtain\looseness=-1
\[
H_{i}=h_{i},\quad i=1,\dots,4,\quad H_{5}=h_{5}+h_{1}t_{4},\quad H_{6}%
=h_{6}+3h_{1}t_{3}+2h_{2}t_{4}+h_{3}t_{5}.
\]

\end{example}

\begin{example}
For the case $n=6$, $m=1$ the formulas (\ref{str}) yield the following matrix
of commutators $\{h_{i},h_{j}\}$:
\[
\left(
\begin{array}
[c]{cccccc}%
0 & 0 & 0 & 0 & 0 & 0\\
0 & 0 & 0 & 0 & 0 & 4h_{1}\\
0 & 0 & 0 & 0 & 2h_{1} & 3h_{2}\\
0 & 0 & 0 & 0 & h_{2} & 2h_{3}\\
0 & 0 & -2h_{1} & -h_{2} & 0 & h_{4}\\
0 & -4h_{1} & -3h_{2} & -2h_{3} & -h_{4} & 0
\end{array}
\right)  ,
\]
so now $d_{c}=1$ while $d_{a}=4$ as in the previous example. Inserting
(\ref{3}) into (\ref{zcr})\ yields%
\begin{align*}
H_{i}  &  =h_{i},\quad i=1,\dots,4,\quad H_{5}=h_{5}+2h_{1}t_{3}+h_{2}%
t_{4},\quad\\
H_{6}  &  =h_{6}+4h_{1}t_{2}+3h_{2}t_{3}+2h_{3}t_{4}+h_{4}t_{5}-\frac{1}%
{2}h_{2}t_{5}^{2}.
\end{align*}

\end{example}

Now turn to the general study of the cases $m=n,n+1$, where
$\mathfrak{g=\mathrm{span}}\left\{  h_{1}\right\}  \mathfrak{\oplus g}_{I_{2}%
}$ (since $I_{1}^{m}$ is empty). The constants $(i_{1},\dots,i_{s+1})$ in
(\ref{2}) are the same as in the previously considered cases up to the sign,
i.e., $(i_{1},\dots,i_{s+1})\rightarrow(-1)^{s}(i_{1},\dots,i_{s+1})$. From
(\ref{str}) it follows that%
\begin{align*}
\text{for}\quad m=n:  &  \quad\mathfrak{g}_{c}=\mathrm{span}\left\{
h_{1}\right\}  ,\quad\mathfrak{g}_{a}=\mathrm{span}\left\{  h_{1}%
,h_{n-k+1},\dots,h_{n}\right\}  ,\\
\text{for}\quad m=n+1:  &  \qquad\ \mathfrak{g}_{c}=\mathrm{span}\left\{
h_{1},h_{n}\right\}  ,\quad\mathfrak{g}_{a}=\mathrm{span}\left\{
h_{1},h_{n-k+1},\dots,h_{n}\right\}  ,
\end{align*}
where $k=\left[  \frac{m}{2}\right]  $. Thus, $d_{c}=\dim\mathfrak{g}_{c}=1$
for $m=n$, $d_{c}=\dim\mathfrak{g}_{c}=2$ for $m=n+1$ while $d_{a}%
=\dim\mathfrak{g}_{a}=k+1$ in both cases. If we now rearrange the Hamiltonians
$h_{i}$ so that $h_{1}^{\prime}\equiv h_{1},$ $h_{i}^{\prime}\equiv h_{n-i+2}$
for $i=2,\dots,n$ we observe that in this ordering the assumptions of Theorem
\ref{main} are all satisfied. Actually, for $m=n+1$ the algebra $\mathfrak{g}$
is nilpotent, while for $m=n$ we have $\left\{  h_{n}^{\prime},h_{j}^{\prime
}\right\}  \in\mathrm{span}\left\{  h_{1}^{\prime},\dots,h_{j}^{\prime
}\right\}  $ which means that $\mathfrak{g}$ is a codimension one extension by
derivation of the nilpotent Lie algebra $\mathfrak{g}_{n-1}$. Thus, we obtain
the following\looseness=-1

\newpage

\begin{corollary}
Suppose that $m=n$ or $m=n+1$ (so $I_{1}^{m}$ is empty). Then the conditions
of Theorem \ref{main} are satisfied and the polynomial-in-times deformation of
$\mathfrak{g}$ given by formulas (\ref{a}) and (\ref{aa}), for the original
ordering of the Hamiltonians $h_{i}$, can be written in the form
\begin{equation}
\hspace*{-4mm}%
\begin{array}
[c]{rcl}%
H_{i} & = & h_{i}-\displaystyle\sum_{r_{1}=i+1}^{n}\left(  \mathrm{ad}%
_{h_{r_{1}}}h_{i}\right)  t_{r_{1}}+\sum_{r_{1}=i+1}^{n}\sum_{r_{2}=r_{1}}%
^{n}\alpha_{ir_{1}r_{2}}\left(  \mathrm{ad}_{h_{r_{2}}}\mathrm{ad}_{h_{r_{1}}%
}h_{i}\right)  t_{r_{1}}t_{r_{2}}\\[7mm]
&  & -\displaystyle\sum_{r_{1}=i+1}^{n}\sum_{r_{2}=r_{1}}^{n}\sum_{r_{3}%
=r_{2}}^{n}\alpha_{ir_{1}r_{2}r_{3}}\left(  \mathrm{ad}_{h_{r_{3}}}%
\mathrm{ad}_{h_{r_{2}}}\mathrm{ad}_{h_{r_{1}}}h_{i}\right)  t_{r_{1}}t_{r_{2}%
}t_{r_{3}}+\cdots
\end{array}
\label{8}%
\end{equation}
where the real constants $\alpha_{ir_{1}\cdots r_{k}}$ can be uniquely
determined from the Frobenius integrability condition (\ref{zcr}).\looseness=-1
\end{corollary}

Note that $d_{c}$ does not enter the above formula; in the case of $m=n+1$
when $d_{c}=2$ the sums in (\ref{8}) end already at $n-1$ since then
$\mathrm{ad}_{h_{n}}=0$ as $h_{n}$ is part of the center of \ the algebra. As
before, the highest order of $t$-polynomials depends on $n$. By analogy with
the previous case, the $\left[  \frac{m}{2}\right]  +1$ Hamiltonians spanning
the Abelian subalgebra $\mathfrak{g}_{a}$, that is, $h_{1}$ and $h_{n-k+1}%
,\dots,h_{n}$ with $k=\left[  \frac{m}{2}\right]  $, are not deformed.

\begin{example}
Consider the case $n=6$, $m=n$. The matrix of commutators $\{h_{i},h_{j}\}$
(\ref{str}) reads%
\[
\left(
\begin{array}
[c]{cccccc}%
0 & 0 & 0 & 0 & 0 & 0\\
0 & 0 & -h_{3} & -2h_{4} & -3h_{5} & -4h_{6}\\
0 & h_{3} & 0 & -h_{5} & -2h_{6} & 0\\
0 & 2h_{4} & h_{5} & 0 & 0 & 0\\
0 & 3h_{5} & 2h_{6} & 0 & 0 & 0\\
0 & 4h_{6} & 0 & 0 & 0 & 0
\end{array}
\right)
\]
and since $k=\left[  \frac{m}{2}\right]  =3$ then the Hamiltonians $h_{i}$
with $i=1,4,5,6$ span an Abelian subalgebra $\mathfrak{g}_{a}$ and are thus
not deformed while $H_{2}$ and $H_{3}$ are found by inserting (\ref{8}) into
(\ref{zcr}) in order to determine the constants $\alpha_{ir_{1}\ldots r_{k}}
$. The result is
\begin{align*}
H_{i}  &  =h_{i},\quad i=1,4,5,6,\\
H_{2}  &  =h_{2}+h_{3}t_{3}+2h_{4}t_{4}+3h_{5}t_{5}+4h_{6}t_{6}+h_{5}%
t_{3}t_{4}+2h_{6}t_{3}t_{5},\\
H_{3}  &  =h_{3}+h_{5}t_{4}+2h_{6}t_{5}.
\end{align*}

\end{example}

\begin{example}
Consider now the case $n=6,~m=n+1$. The matrix of commutators $\{h_{i}%
,h_{j}\}$ reads now
\[
\left(
\begin{array}
[c]{cccccc}%
0 & 0 & 0 & 0 & 0 & 0\\
0 & 0 & -h_{4} & -2h_{5} & -3h_{6} & 0\\
0 & h_{4} & 0 & -h_{6} & 0 & 0\\
0 & 2h_{5} & h_{6} & 0 & 0 & 0\\
0 & 3h_{6} & 0 & 0 & 0 & 0\\
0 & 0 & 0 & 0 & 0 & 0
\end{array}
\right)
\]
Again, $h_{i}$ are not deformed for $i=1,4,5,6$ while $H_{2}$ and $H_{3}$ are
found as usual by inserting (\ref{8}) into (\ref{zcr}). The result is%
\[
H_{i}=h_{i},\ i=1,4,5,6,\ H_{2}=h_{2}+h_{4}t_{3}+2h_{5}t_{4}+3h_{6}%
t_{5}-\displaystyle\frac{1}{2}h_{6}t_{3}^{2},\ H_{3}=h_{3}+h_{6}t_{4}.
\]

\end{example}

Finally, let us again return to the general theory and find integrable
deformations of our algebra $\mathfrak{g}$ in the case $1<m<n$. In this case
both components $\mathfrak{g}_{I_{1}}$ and $\mathfrak{g}_{I_{2}}$ in the
splitting $\mathfrak{g=g}_{I_{1}}\mathfrak{\oplus g}_{I_{2}}$ are nontrivial,
with $\dim\mathfrak{g}_{I_{1}}=n-m+1$ and $\dim\mathfrak{g}_{I_{2}}=m-1$. Each
of the components has an Abelian subalgebra of its own. We denote the Abelian
subalgebras of $\mathfrak{g}_{I_{1}}$ and $\mathfrak{g}_{I_{2}}$ by
$\mathfrak{g}_{a_{1}}$ and $\mathfrak{g}_{a_{2}}$, respectively, with (compare
this with (\ref{dimga}))
\begin{align*}
\dim\mathfrak{g}_{a_{1}}  &  =\left[  \frac{n+3-m}{2}\right]  \equiv d_{a_{1}
}\\
\dim\mathfrak{g}_{a_{2}}  &  =\left[  \frac{m}{2}\right]  \equiv d_{a_{2}},
\end{align*}
so
\[
\dim\mathfrak{g}_{a}=\dim\mathfrak{g}_{a_{1}}+\dim\mathfrak{g}_{a_{2}};
\]
$\mathfrak{g}_{a_{1}}$ and $\mathfrak{g}_{a_{2}}$ are given by
\begin{align*}
\mathfrak{g}_{a_{1}}  &  =\mathrm{span}\left\{  h_{1},h_{2},\ldots
,h_{d_{a_{1}}}\right\}  \text{ }\\
\mathfrak{g}_{a_{2}}  &  =\mathrm{span}\left\{  h_{n-d_{a_{2}}+1},\ldots
,h_{n}\right\}
\end{align*}
so
\[
\mathfrak{g}_{a}=\mathrm{span}\left\{  h_{1},h_{2},\ldots,h_{d_{a_{1}}%
},h_{n-d_{a_{2}}+1},\ldots,h_{n}\right\}  .
\]
Therefore, the Hamiltonians $h_{d_{a_{1}}+1},\dots,h_{n-m+1}$ belonging to
$\mathfrak{g}_{I_{1}}$ should then be deformed by formulas (\ref{3}) with
$d_{c}=1$ (for $m=n+1$ the center is two-dimensional, spanned by
\thinspace$h_{1}$ and $h_{n}$, but $h_{n}$ does not belong to $\mathfrak{g}%
_{I_{1}}$) while the Hamiltonians $h_{1},\dots,h_{d_{a_{1}}}$ remain
unchanged. Likewise, the Hamiltonians $h_{n-m+2},\dots,h_{n-d_{a_{2}}}$ should
be deformed according to (\ref{8}) while the last $d_{a_{2}}$ Hamiltonians
remain unchanged.\looseness=-1

\begin{example}
Consider the case $n=11$, $m=6$. The matrix of commutators $\left\{
h_{i},h_{j}\right\}  $ is
\[
\left(
\begin{array}
[c]{ccccccccccc}%
0 & 0 & 0 & 0 & 0 & 0 & 0 & 0 & 0 & 0 & 0\\
0 & 0 & 0 & 0 & 0 & 4h_{1} & 0 & 0 & 0 & 0 & 0\\
0 & 0 & 0 & 0 & 2h_{1} & 3h_{2} & 0 & 0 & 0 & 0 & 0\\
0 & 0 & 0 & 0 & h_{2} & 2h_{3} & 0 & 0 & 0 & 0 & 0\\
0 & 0 & -2h_{1} & -h_{2} & 0 & h_{4} & 0 & 0 & 0 & 0 & 0\\
0 & -4h_{1} & -3h_{2} & -2h_{4} & -h_{4} & 0 & 0 & 0 & 0 & 0 & 0\\
0 & 0 & 0 & 0 & 0 & 0 & 0 & -h_{8} & -2h_{9} & -3h_{10} & -4h_{11}\\
0 & 0 & 0 & 0 & 0 & 0 & h_{8} & 0 & -h_{10} & -2h_{11} & 0\\
0 & 0 & 0 & 0 & 0 & 0 & 2h_{9} & h_{10} & 0 & 0 & 0\\
0 & 0 & 0 & 0 & 0 & 0 & 3h_{10} & 2h_{11} & 0 & 0 & 0\\
0 & 0 & 0 & 0 & 0 & 0 & 4h_{11} & 0 & 0 & 0 & 0
\end{array}
\right)
\]
If we perform the deformation on each subalgebra separately, as described
above, we obtain
\begin{align*}
H_{i}  &  =h_{i},\quad i=1,\dots,4,9,\dots,11,\quad H_{5}=h_{5}+h_{2}%
t_{4}+2h_{1}t_{3},\\
H_{6}  &  =h_{6}+4h_{1}t_{2}+3h_{2}t_{3}+2h_{3}t_{4}+h_{4}t_{5}-\frac{1}%
{2}h_{2}t_{5}^{2},\\
H_{7}  &  =h_{7}+h_{8}t_{8}+2h_{9}t_{9}+3h_{10}t_{10}+4h_{11}t_{11}%
+h_{10}t_{8}t_{9}+2h_{11}t_{8}t_{10},\\
H_{8}  &  =h_{8}+h_{10}t_{9}+2h_{11}t_{10}.
\end{align*}

\end{example}

\section*{Acknowledgments}

AS would like to thank R. Popovych and P. Zusmanovich for helpful comments.

The research of AS and MB, as well as the visit of MB to Opava in November
2017, were supported in part by the Grant Agency of the Czech Republic (GA
\v{C}R) under grant P201/12/G028. The research of AS was also supported in
part by the Ministry of Education, Youth and Sports of the Czech Republic
(M\v{S}MT \v{C}R) under RVO funding for I\v{C}47813059.

\end{document}